\documentclass{amsart}
\usepackage{amstext, amssymb, amsthm, amsfonts, latexsym,bbm,tikz}
\usepackage{mathtools}
\usetikzlibrary{matrix}
\usepackage{varioref}
\usepackage{float}
\usepackage{longtable}
\usepackage{todonotes}
\usepackage{bbm}
\usepackage[utf8]{inputenc}
\usepackage[english]{babel}
\usepackage{booktabs}
\usepackage{caption}
\usepackage{subcaption}
\usepackage{graphicx}
 \usepackage{enumitem}
 \setenumerate[1]{label=(\Roman*)}
 
\usepackage{comment}
\usepackage{tikz-cd}
\usepackage{hyperref}
\usepackage[version=3]{mhchem}
\usepackage{booktabs} 
\usepackage[left=1in, top=1in, right=1in, bottom=1in]{geometry}

\newtheorem{theorem}{Theorem}[section]
\newtheorem{lemma}[theorem]{Lemma}

\theoremstyle{definition}

\theoremstyle{remark}

\theoremstyle{Conjecture/open problem}

\theoremstyle{assumption}
\newtheorem{assumption}{Assumption}

\theoremstyle{conjecture}

\def\1{\mathbf{1}}

\def\p{^{\prime}}

\def\E{\mathbb{E}}

\def\N{\mathbb{N}}

\numberwithin{equation}{section}

\begin{document}

\title{On misconceptions about the Brier score in binary prediction models
}


\author[Linard Hoessly]{Linard Hoessly}

\address{Data Center of the Swiss Transplant Cohort Study, University Basel \& university hospital Basel, Basel, 4031 Switzerland}
\email{linard.hoessly@hotmail.com}

\date{}


\date{}

\begin{abstract}
The Brier score is a widely used metric evaluating overall performance of probabilistic predictions for binary outcomes in clinical research. However, its interpretation can be complex, as it does not align with commonly taught concepts in medical statistics. Consequently, the Brier score is often misinterpreted, sometimes to a significant extent, a fact that has not been adequately addressed in the literature. We aim to explore prevalent misconceptions surrounding the Brier score and elucidate the understanding and interpretation of Brier scores for practitioners.
 \end{abstract}

\maketitle

\section{Introduction: What is the Brier score?}\label{intro}

The Brier score \cite{Brier_original} is a widely used metric evaluating the accuracy of probabilistic predictions in binary outcomes for clinical research \cite{steyerberg2019clinical,harrell2015regression}. It assesses the overall performance of prediction models that estimate the likelihood of medical outcomes like disease progression or treatment response.

 Given probabilistic predictions $p_i$ and observed outcomes $y_i$, the Brier score is defined as:
\begin{equation}\label{Brier} BS(p, y) = \frac{1}{n} \sum_{i=1}^{n} (p_i - y_i)^2. \end{equation}

where:
\begin{itemize}
\item $n$ is the total number of predictions and observations,
\item $p_i$ represents the predicted probability of an event occurring for the $i$-th case (e.g., the probability of a patient developing a condition),
\item $y_i$ is the actual observed outcome (coded as 1 if the event occurred, 0 if it did not).
\end{itemize}

Typically, the $y_i$ in such prediction models are assumed to be realisations of independent Bernoulli random variables $Y_i \sim Bern(q_i)$, where $q_i \in [0,1]$ \cite{harrell2015regression,Rufibach2010}. A Bernoulli random variable is like a potentially biased coin flip: each patient either has the event 1 or 0, based on their individual risk. Two random variables are independent if the probability of an outcome for one is unaffected by the outcome of the other \cite{wasserman2010statistics}. Correspondingly the best i-th prediction that can be obtained is the true underlying probability, i.e, $p_i=q_i$.

Brier scores offer a comprehensive evaluation of probabilistic predictions and are strictly proper \cite[Theorem 1]{Byrne_auc}, meaning that in expectation it is minimised if and only if the predictions are the true probabilities \cite{scoring_rule}. Note the distinction between accurate probabilistic predictions and clinical usefulness: while clinicians may prefer binary classifications, the quality of a probability prediction is judged by how closely its probabilities reflect the true risks. The Brier score is an evaluation measure that quantifies this closeness of predictions to true risks.
If desired, predicted probabilities can be assessed for clinical impact, e.g., via net benefit \cite{Steyerberg_Vickers_Cook_Gerds_Gonen_Obuchowski_Pencina_Kattan_2010}, or used to derive a classification \cite{hand_classification}.

Despite its widespread use \cite{Steyerberg2010_scaled_2,Mackillop1997,Redelmeier1991}, the Brier score is potentially often misunderstood in clinical research. Unlike more familiar statistical notions, it does not fit neatly into traditional statistical concepts potentially commonly taught in medical education \cite{illowsky2013introductory,armitage2001statistical,motulsky2010intuitive,rothman2008modern}.
Furthermore, the Brier score is mathematically equivalent to the mean squared error (MSE), a concept introduced by C.F. Gauss \cite{gauss1957gauss}. The MSE is widely applied in ordinary least squares regression \cite[$\S$ 11.3.1]{CaseBerg}, statistical learning \cite[$\S$ 2.2.1]{James2013}, or machine learning evaluation \cite{Flach2019}. The connection of Brier score and MSE can also lead to confusion. Brier score and MSE are used in different contexts, as the Brier score compares a probability to an outcome of the binary random variable in the sense of scoring rules \cite{scoring_rule}, while the MSE usually compares two real continuous values, in statistics typically comparing an estimator to the true value \cite{CaseBerg}. In particular, misconceptions about the Brier score are not uncommon and can sometimes be reinforced by potentially misleading statements in the literature \cite{Steyerberg2001,Steyerberg_Vickers_Cook_Gerds_Gonen_Obuchowski_Pencina_Kattan_2010,steyerberg2019clinical,cox_prediction_comp,van_Gelovene069249,Carriero2025}. Given the importance of accurate interpretation in clinical applications, it is crucial to address these misunderstandings.

Evaluation of binary prediction models has been widely studied in the medical literature. A review of traditional and modern performance measures is given in \cite{Steyerberg_Vickers_Cook_Gerds_Gonen_Obuchowski_Pencina_Kattan_2010}, which are also discussed in Harrells book \cite[$\S$ 10]{harrell2015regression} for regression models. Steyerberg's book \cite{steyerberg2019clinical} offers a comprehensive guide to clinical model development and validation. Alternative classification metrics such as the Gini coefficient and Pietra index were examined in \cite{Gini_Auroc}. Other works explore alternative evaluation scores \cite{alternative_scores} and compare metrics like the Brier score with net benefit analysis \cite{assel_brier_2017}. The Brier score has also been decomposed for deeper insights \cite{Yates1982}, and remains relevant in AI-based medical prediction \cite{vancalster2024performanceevaluationpredictiveai} or survival outcome evaluation \cite{cox_prediction_comp}.

We aim to clarify common misunderstandings about the Brier score, explain why they arise, and provide guidance on its appropriate interpretation in clinical research. As it is also directed at researchers in medicine, we first introduce some notions, and give a summary of the misconceptions. Then we go through the misconceptions and end with a conclusion. Supplementary analysis is kept in the Appendix for the interested reader.\newpage
\subsection*{Key terms we will use and what they mean}
\begin{itemize}
\item \textbf{Random variable:} A random variable is a quantity that depends on the outcome of a random process. 
For example, let \( Y \) denote whether a leaving patient is readmitted within 30 days:
\[
Y = 
\begin{cases}
1 & \text{if readmitted} \\
0 & \text{if not readmitted}
\end{cases}
\]
Before observation, \( Y \) is unknown and varies due to chance. If $Y \sim Bern(q)$ for $q=0.2$, the probability to observe a $1$ is $20\%$, and the probability to observe a $0$ is $80 \%$. 

\item \textbf{Expectation:} Expectation is a way to describe the average outcome we expect over the long run, if we repeat a situation many times under the same conditions, denoted by $\E(\cdot)$. The expectation of $Y \sim Bern(q)$ for $q=0.2$, $\E(Y)$, is $0.2$.
\item \textbf{Perfect prediction:} In clinical prediction models, a perfect prediction means that the predicted probabilities of outcomes exactly match the true underlying risk for each individual patient. In the case of our setting of the Brier score, if the true probabilities are given by $q=(q_1,\cdots ,q_n)$, the perfect prediction is given by $p=(p_1,\cdots ,p_n)=(q_1,\cdots ,q_n)$.

As an example, suppose a model estimates that a patient has a 20\% chance of being readmitted to the hospital within 30 days. If that patient's true risk is  exactly 20\%, then the model has made a perfect prediction for that individual.
\end{itemize}
\subsection*{Quick reference on common misconceptions}\hfill\break
\begin{table}[H]
\centering
\renewcommand{\arraystretch}{1.4} 
\begin{tabular}{p{6.2cm} p{8.3cm}}
\toprule
\textbf{Misconception} & \textbf{Reality} \\
\midrule
\# 1: Brier score of 0 = perfect model &
A Brier score of 0 implies extreme predictions (0\% or 100\%) that exactly match outcomes. This is odd in practice and typically indicates errors.\\

\# 2: Lower Brier score always means a better model &
A lower Brier score can be misleading across datasets with different underlying distributions for the outcomes. It is only meaningful to compare Brier scores within the same population and context. 
\\

\# 3: A low Brier score indicates good calibration &
Calibration and Brier score measure different aspects. Calibration refers to how well predicted probabilities reflect observed risks; a model can have a low Brier score and still be poorly calibrated. \\

\# 4: A Brier score near $\bar{y} - \bar{y}^2$ means the model is useless &
Even perfect predictions can yield a Brier score near $\bar{y} - \bar{y}^2$ if the true risks are close to the mean incidence. This does not necessarily imply non-informativeness. \\

\# 5: Brier score cannot exceed $\bar{y} - \bar{y}^2$ for reasonable predictions &
As a realisation of a random variable, the Brier score can exceed the threshold due to chance or reasonable predictions. \\
\bottomrule
\end{tabular}
\caption{Common misconceptions about the Brier score and how they contrast with statistical reality.}
\end{table}

\subsection*{Related literature}
Some of our points have been previously observed. We review related references that observe similar findings. However, given the widespread use of the Brier score, our literature review is necessarily partial. \cite{jewson2004problembrierscore} outlines examples where the model comparison of expected BSs is potentially contrary to how a human would judge. \cite{Rufibach2010} outlined the distinction between calibration and prediction error, emphasizing the misunderstanding that low Brier score indicates good calibration, while \cite{ML_eval_21} notes that Brier score comparisons across datasets should be avoided as it depends on the incidence rate. 
\subsection*{Acknowledgements}
We thank Lucia de Andres Bragado, Julien Vionnet, Simon Schwab, and Matthew Parry for helpful discussions and feedback.

\section{Main properties of the Brier score}
The Brier score in \eqref{Brier} is a measure to quantify the accuracy of probabilistic predictions, taking values between 0 and 1 with lower values indicating better performance. As the $y_i$ are realisations of random variables, the Brier score is a random variable. Hence any evaluation of the Brier score has a random component. It is strictly proper, meaning in expectations the perfect prediction uniquely minimises the Brier score of \eqref{Brier}. Even more holds: in expectation the Brier score preserves the  Euclidean distance order $l_2$ between predictions $p\in[0,1]^n$ to the true probability $q\in[0,1]^n$, where $l_2(p,q)=\sqrt{\sum_{i=1}^n(p_i-q_i)^2}$ \cite{nau_effective}.

Suppose we have two sets of predicted probabilities, $p$ and $p'$, from two different models, and let $q$ represent the true probabilities. If $p$ is closer to $q$ than $p'$ is in terms of Euclidean distance, i.e.,
$$
l_2(p, q) < l_2(p', q),
$$
then the expected Brier score for $p$ will be lower than that of $p'$.

To build intuition, consider a 2-dimensional example where the true probabilities are $q = (1/2, 1/2)$. In this case, any prediction falling inside a circle centered at \( (1/2, 1/2) \) will have a lower expected Brier score than one on or outside the circle. This provides a geometric view of probabilistic accuracy: the closer your predictions are to the truth (in Euclidean distance), the better the model performs.
\begin{center}
\begin{tikzpicture}[scale=3]
  \draw[very thick, gray] (0,0) rectangle (1,1);

  \draw[->, very thick] (0,0) -- (1.1,0) node[right] {\LARGE$x$};
  \draw[->, very thick] (0,0) -- (0,1.1) node[above] {\LARGE$y$};

  \draw[very thick] (1,0.02) -- (1,-0.02) node[below] {\LARGE$1$};
  \draw[very thick] (0.02,1) -- (-0.02,1) node[left] {\LARGE$1$};

  \draw[blue, ultra thick] (0.5,0.5) circle [radius=0.2];

  \draw[red, thick, ->] (0.5,0.5) -- (0.7,0.5) node[midway, above] {$0.2$};

  \filldraw[red] (0.5,0.5) circle (0.008) node[below] {\scriptsize$(\frac{1}{2},\frac{1}{2})$};
\end{tikzpicture}
\end{center}

We can further quantify the expected behavior when slightly perturbing the predicted probabilities away from the true values by $\varepsilon$. This results in an expected Brier score increase of $\varepsilon^2$. To be more specific, if, say, the predicted probabilities are $0.1$ more off, the expected Brier score will be $0.1^2=0.01$ bigger. In contrast, changes in the true probability can lead to more significant shifts in expectation.

While individual true probabilities are unobservable in practice, the observed prevalence provides an estimate of the expected prevalence and can serve as a non-informative reference point for comparison. For idealised scenarios, the law of large numbers or central limit theorem can be used to argue that the observed score is a reliable estimate of its expected value. For large datasets, this can be used to justify treating the empirical Brier score as a stable summary measure of model performance.  To summarise and simplify our previous points, we note that an observed Brier score is a function of
\begin{enumerate}
\item the underlying true probabilities (the $q_i$s),
\item the closeness of the predictions when compared to the true probabilities (how close $p_i$ is from $q_i$),
\item some randomness that comes from the Bernoulli random variables (the observed $y_i$ that are realisations of $Y_i\sim Bern(q_i)$).
\end{enumerate}
The influence of the randomness can be expected to decreases with $n$, and the expectation of the Brier score can be seen as the long term average, which represent typical values if $n$ is big.
More details on properties as well as calculations for the derivations above can be found in Appendix $\S$ \ref{about_BS}, and connections to other measures in Appendix $\S$ \ref{connections}.

\subsection{Means of understanding the Brier score}  

We will use the following approaches to better understand the Brier score \eqref{Brier}:  
\begin{itemize}  
      \item \textbf{Expectation of the Brier score:}  
    We will analyze the expected value the Brier score \eqref{Brier} takes, which will help to illustrate typical observed values of the Brier score.  

    \item \textbf{Simulation-Based Evaluations:}  
    We assess observations of the Brier score by simulating Bernoulli outcomes under different sampled probability distributions for \( q_i \) and different sample sizes. More detail is in Appendix $\S$ \ref{ADEMP}.
\begin{itemize}
    \item \textbf{Sample size $n$:}. Settings considered: $n\in\{300,1000\}$.
    \item \textbf{True distribution}: We consider $q_i$ as realisations of random variables. The $q_i\in[0,1]$ are then used to simulate $Y_i\sim Bern(q_i)$. 
    \item \textbf{Predictor distribution}: The $p_i\in[0,1]$ used in \eqref{Brier} are considered as functions of $q_i$ with potentially random error.
\item \textbf{Estimand}: Includes median Brier score, $5\%,95\%$ quantiles, and the distribution via violin plots.
\end{itemize}
\end{itemize}
\section{Misconceptions}\label{misconceptions} Below are the most common misinterpretations of the Brier score when evaluating probability predictions for binary events, accompanied by an explanation and examples illustrating why it is incorrect. 

\subsection{Misconception \#1: A Brier score of 0 means a perfect model, and a perfect model has Brier score 0.}\label{subs_misc1}
\begin{itemize}
\item \textbf{A Brier score of 0 means a perfect model.} A Brier score of 0 implies perfect alignment between predicted probabilities and observed outcomes, with predictions exclusively 0 or 1. The true probabilities are within $[0,1]$, usually expected in $(0,1)$.  Hence, rather than signalling a perfect model, an observed Brier score of 0 potentially indicates errors. 
\item\textbf{A perfect model has Brier score 0.} With at least one of the true probabilities $q_i$ in $(0,1)$, observing a Brier score of 0 with perfect predictions is impossible (see Appendix \ref{math_proof}). Hence in normal situations perfect models will have a Brier score bigger 0.
\end{itemize}

To illustrate this and for later reference, we present simulations with perfect predictions, showing that ideal models do not yield a Brier score of 0. In fact, expected Brier scores for a perfect model can be notably high. 
\begin{figure}[ht]
    \centering
    \includegraphics[width=0.7\textwidth]{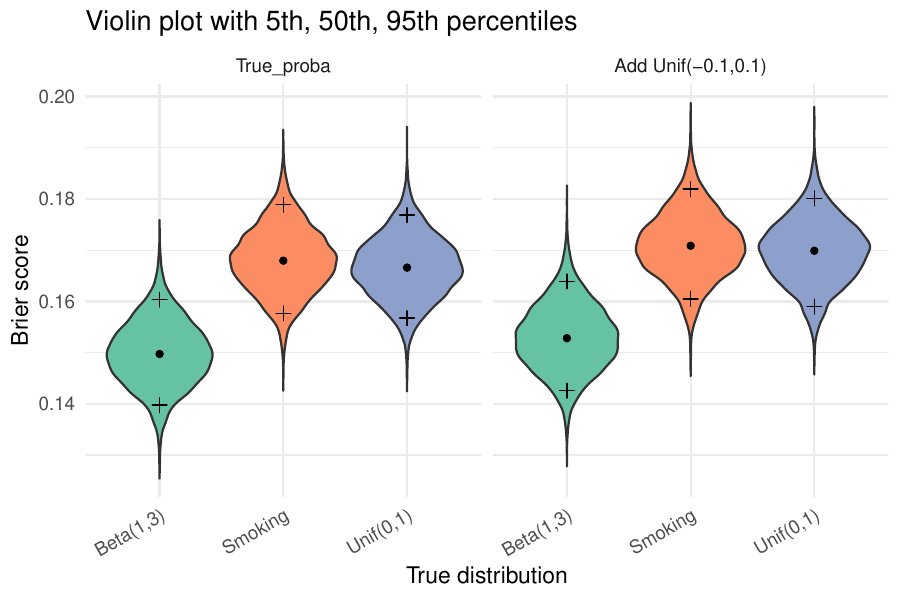}
    \caption{Violin plot for simulations with $n=1000$. The left side has perfect predictions, the right $Unif(-0.1,0.1)$ noise added.}
    \label{tab:distribution_comparison}
\end{figure}


\subsection{Misconception  \#2: When comparing two prediction models, the model with lower Brier score is better.}  
Brier scores enable model comparison across datasets, but comparisons can be misleading. We examine three scenarios of increasing risk to illustrate this point.

\subsubsection{Comparing Models on the same dataset via Brier score:}  
When evaluating two models on the same dataset, lower Brier score indicate better fit. There are two potential issues that can prevent the better model to have lower Brier score:
    \begin{enumerate}
    \item It is possible that the worse model has a better score by chance. However, with higher sample sizes this becomes more unlikely.
    \item In expectation the Brier score ranks according to the $l_2$-distance between prediction $p\in[0,1]^n$ to the true probability $q\in[0,1]^n$. Changing perfect predictions from $q_i$ to $q_i+\varepsilon$ or $q_i-\varepsilon$ gives the same result in expectation. However, for humans the direction can matter particularly for $q_i$ close to zero or one \cite{jewson2004problembrierscore}. Consider the example with one observation from \cite[$\S$ 3]{jewson2004problembrierscore} where the probability of an event is \( P(Y_1=1) = \frac{1}{10} \) (i.e., 10\% chance), and compare two models:  
    \begin{itemize}  
        \item Model 1: \( p_1 = 0 \) (predicting the event will never occur)  
        \item Model 2: \( \tilde{p_1} = \frac{1}{4} \) (predicting the event occurs with a probability of 25\%)  
    \end{itemize}  
Model 1 has an expected Brier score of $0.1$, as compared to $0.1125$ for model 2. However, model 1 predicts a zero probability for the observation that has actually a 10\% probability, hence 25\% might seem better from a humans perspective. 
    \end{enumerate}

\subsubsection{Comparing Brier scores across different datasets with similar incidences (also known as class imbalance) is meaningful.}
Some recommend to compare Brier score of prediction models on datasets with similar class imbalances \cite{Carriero2025}.
 However, the true probability distribution is unobservable, and true values $q_i$ strongly influences the expectation of the Brier score. Thus, even with identical class imbalances, we cannot assume that the true probability distributions are comparable and the same caution as in the previous point should be exercised in interpreting such comparison. As an example, compare three cases: True distribution always $1/2$, independent $Unif(0,1)$ from Figure \ref{tab:distribution_comparison}, or a $0$ or $1$ distribution with each 50 \%. The perfect predictions have in each case expected Brier score $0.25, 0.16$ and $0$.  Hence perfect predictions for true distribution always $1/2$ and expected Brier score $0.25$ is better than, say predictions with noise $Unif(-0.1,0.1)$ for true probabilities $Unif(0,1)$ from Figure \ref{tab:distribution_comparison} with expected Brier score $0.17$. Potentially, similarities in outcome and population trough population characteristics might indicate how true probabilities differ or align in a clinical prediction model setting.
  
 \subsubsection{Comparing Brier scores across datasets with different incidences:}  
If true outcome distributions differ, even perfect models yield different Brier score distributions and (potentially) expectations. Thus, cross-dataset comparisons may lack meaningful insight. Observing different incidences may indicate that the true outcome distributions differ, making the Brier score comparison unreliable. Compare, e.g. the $Beta(1,3)$ prediction with prediction error $+Unif(-0.1,0.1)$ to the perfect smoke prediction model from figure \ref{tab:distribution_comparison}. The $Beta(1,3)$ prediction with prediction error has lower Brier score mostly when compared to perfect smoking prediction, nonetheless the perfect model is obviously better.

\subsection{Misconception \#3: A low Brier score indicates good calibration.}\label{detailed}
A low Brier score does not necessarily indicate good calibration of a model. We can have perfect predictions, but low or high Brier score due to the underlying probabilities, or similary not so accurate or biased predictions but a Brier score that is low or high due to the underlying probabilities. In expectation a change of perfect prediction from $p_i=q_i$ to $q_i+\varepsilon$ or $q_i-\varepsilon$ gives the same result (cf., e.g., \eqref{Expectation_diff2}), and miscalibration where errors tend to go mostly in one direction are equally punished, but e.g. calibration is differently affected as illustrated in Figure \ref{tab:distribution_comparison_CIL}.

Assessing calibration should be done using additional metrics like calibration in the large (CIL), calibration curves or similar evaluation components \cite{ICI_co,van_calster_calibration_2019}. Hence Brier score should not be the sole criterion for evaluating model performance.

\begin{figure}[ht]
    \centering
    \includegraphics[width=0.7\textwidth]{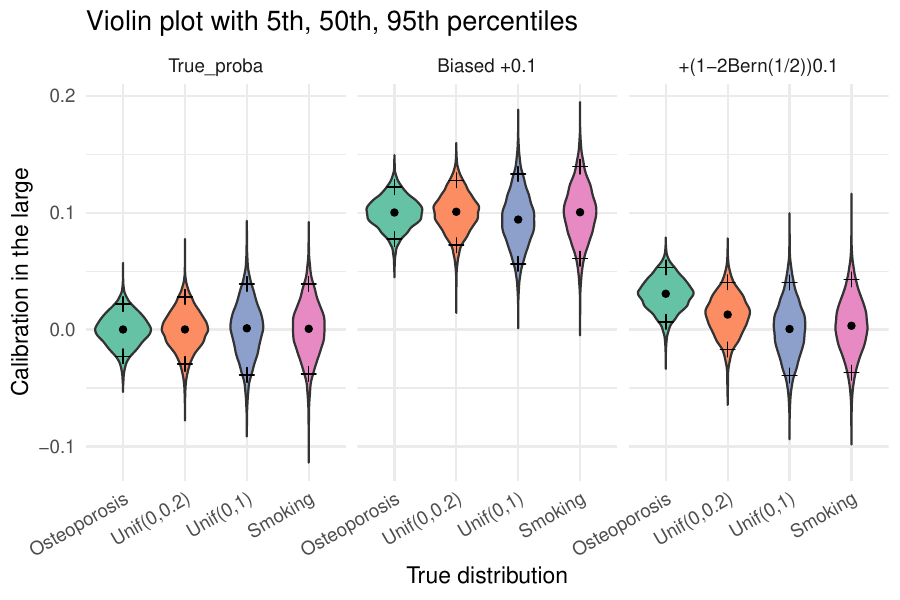}
    \caption{Violin plot for simulations with $n=300$.}
    \label{tab:distribution_comparison_CIL}
\end{figure}

\subsection{Misconception \#4: Having a Brier score of around $\bar{y}-\bar{y}^2$ where $\bar{y}$ is the mean observed incidence means we have a useless or non-informative model. As an example, if $\bar{y}=0.5$, then $\bar{y}-\bar{y}^2=0.25$, or , if $\bar{y}=0.1$, then $\bar{y}-\bar{y}^2=0.090$}\label{50_prev}
Note that the lowest expected Brier score occurs when predicted probabilities match the true probabilities, representing a perfect prediction that cannot be improved. However, we cannot observe the true probabilities. With observed incidences of $0.5$, perfect predictions can yield a Brier score of close to 0.25. Hence if we observe a Brier score of around $\bar{y}-\bar{y}^2$, the following alternatives to a bad model could explain such a Brier score:
\begin{itemize}
\item Many of the true probabilities are around $\bar{y}$, making expected Brier scores of perfect predictions close to $\bar{y}-\bar{y}^2$.
\item For $n$ low, randomness can make the Brier score higher than its expectation.
\end{itemize}

As an illustration of example values, consider $\bar{y}-\bar{y}^2-BS_{perf}$ in figure \ref{tab:distribution_comparison2} on the left, where $BS_{perf}$ is the Brier score under perfect predictions. The observed median is very low, with $5\%$ percentile below or around zero.\color{black}

\begin{figure}[ht]
    \centering
    \includegraphics[width=0.7\textwidth]{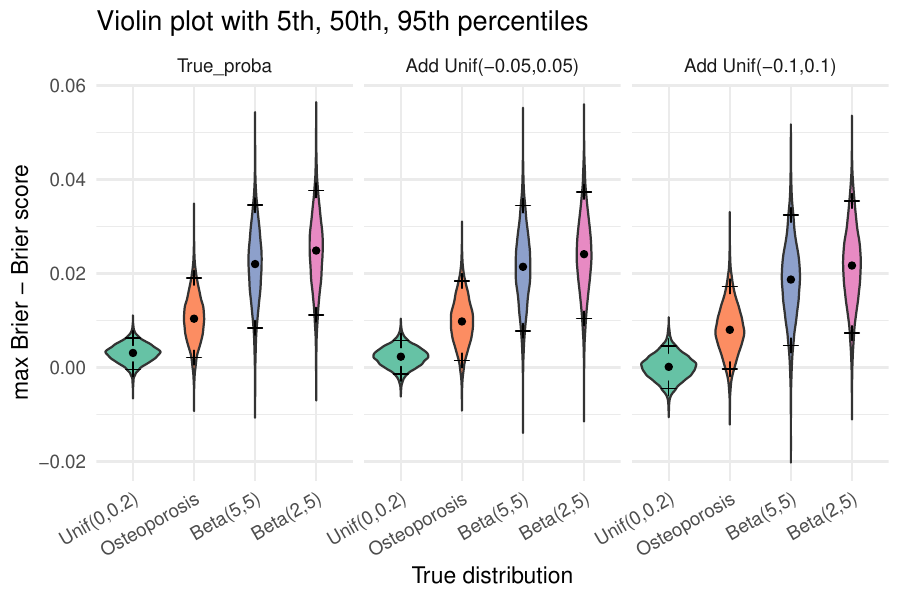}
    \caption{Violin plot for simulations with $n=300$.}
    \label{tab:distribution_comparison2}
\end{figure}

\subsection{Misconception \#5: For an observed incidence of $\bar{y}$, the Brier score \eqref{Brier} for reasonable predictions can not be bigger than $\bar{y}-\bar{y}^2$ }\label{prev_limit}

Although the observed Brier score might often be bounded above by $\bar{y} - \bar{y}^2$ in practice, higher values can occur. Since true probabilities $q_i$ are unobservable, we cannot rule out that even perfect predictions yield a score near the unobservable $\bar{q}-\bar{q}^2$, where $\bar{q}$ is the true mean incidence $\bar{q} = \frac{1}{n}\sum_{i=1}^n q_i$. The Brier score's randomness stems from the outcomes, also making $\bar{y}$ random. 
Hence a Brier score exceeding $\bar{y} - \bar{y}^2$ may also result from the same options as given in misconception \#4.
To illustrate this, we show for perfect predictions in some settings with $n=300$, the probability of having a Brier score that is bigger than $\bar{y}-\bar{y}^2$ is nonzero. As in practice we will not have perfect predictions, for reasonable predictions the probability to have a Brier score bigger than $\bar{y}-\bar{y}^2$ is higher.,The right part can be understood as lower bounds for corresponding probabilities in the corresponding settings. 
\begin{figure}[ht]
    \centering
    \includegraphics[width=0.5\textwidth]{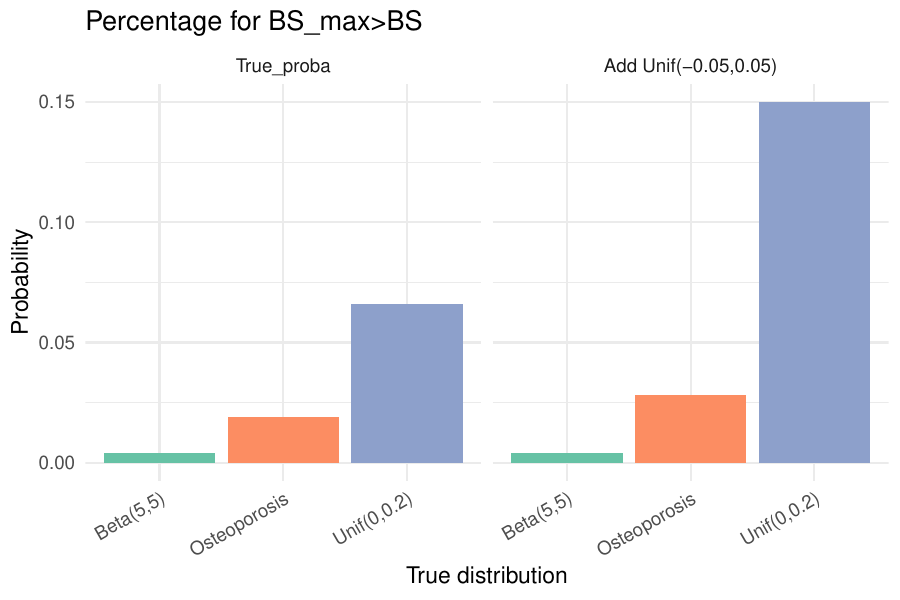}
    \caption{Histogram for simulations with $n=300$.}
    \label{tab:distribution_comparison4}
\end{figure}

\section{Conclusions and final remarks}\label{conclusions}
We adressed common misconceptions regarding the interpretation and use of the Brier score. The Brier score is a realisation of a random variable. The true underlying probabilities influence the expectation of Brier score strongly, potentially stronger than the closeness of the predictions to the true probabilities. Brier scores of zero are rather an indication of errors in realistic settings, and a low Brier score does not necessarily indicate a perfect model. Comparisons of Brier scores across models on the same data can be done, and across different data should be avoided or interpreted carefully. Low Brier scores do not guarantee good calibration as  evaluating calibration should be done with other metrics. In analogy to idealised settings, randomness in observed Brier scores can be expected to decrease with bigger sample sizes. A recent literature review on clinical prediction models found that sample sizes used had median sample size of $1250$ with $(Q1,Q3)=(353,188860)$\cite{Christodoulou2019}. Hence at least a quarter of the prediction models had relatively low sample sizes, with randomness similar to the settings in our simulations with $n=300$ or $n=1000$. 

The Brier score remains a valuable metric for assessing probabilistic predictions, but its interpretation requires an understanding of how the underlying probabilities and closeness of predictions influence the observed value. Once misconceptions are avoided, the Brier score serves as a reliable relative measure of overall performance. In particular, it is effective and strictly proper and hence in expectation, it reflects Euclidean distance between predictions, and is minimised uniquely at the true probabilities making relative comparisons on the same data meaningful.

 \bibliographystyle{plain}

 \bibliography{pred_references}

\appendix
\section{More on the Brier score}\label{about_BS}
In this section, we review key properties of the Brier score, delve into some mathematical points, and give more details on the simulations we use to visualise the behavior of the Brier score.

When evaluating the Brier score \eqref{Brier}, one can compare \eqref{Brier} to the trivial Brier score when entering $p_{1/2}=(1/2,\cdots,1/2)$, which, independent of the specific value of $y$ gives $BS(p_{1/2},y)=1/4$. A slightly better prediction is the mean incidence $\bar{y}_v=(\bar{y},\cdots,\bar{y})$, which when used in the Brier score gives
$$
BS(\bar{y}_v,y)= \bar{y}-\bar{y}^2.
$$
Clearly, $BS(p_{1/2},y)=1/4\geq BS(\bar{y}_v,y)= \bar{y}-\bar{y}^2$. For  high or  low incidences, $BS(\bar{y}_v,y)$ is low, and, e.g., for $\bar{y}=0.1$ or  $\bar{y}=0.9$, $BS(\bar{y}_v,y)=0.09$. Furthermore, $\bar{y}-\bar{y}^2$ is symmetric around $0.5$, see Figure \ref{fig:best_exp}.

\subsection{Key properties about the Brier score} \label{prop_BS} 
We summarise essential properties of the Brier score.  
\begin{enumerate}  [label=\textbf{(B\Roman*)}]
    \item\label{BS_range} \textbf{Range and interpretation:} The Brier score takes values in the interval $[0,1]$, with lower values typically indicating more accurate probabilistic predictions.
    \item\label{BS_random} \textbf{The Brier score is a random variable:} The Brier score is a function of random variables and as such it is also a random variable. In clinical prediction models, the outcomes \( y_i \) are realisations of \( Y_i \sim \text{Bern}(q_i) \), where \( q_i \in [0,1] \).

    \item\label{BS_optimal} \textbf{Optimal predictions and true probabilities:} The unique optimal prediction minimising the expected Brier score is the true outcome probability, i.e., for all $i$ $p_i=q_i$, as the Brier score is strictly proper \cite[Theorem 1]{Byrne_auc}.  
        \item\label{BS_1d} \textbf{Expectation of Brier score:} 
        \begin{itemize}
        \item $(n=1)$: As an illustrative example consider the case $n=1$. Using the true probability, we calculate the expectation. Let $Y_1\sim Bern(q_1),q_1\in[0,1]$ with prediction $p_1\in[0,1]$. Then the expectation of Brier score is 
     \begin{equation}\label{Expectation_Brier}
   g(p_1,q_1) := \E_{Y_1\sim Bern(q_1)}[BS(p_1, Y_1)]=p_1^2 - 2p_1q_1 + q_1.
\end{equation}
For the interested reader, the calculation is given in Appendix \ref{exp_BS}. The optimal prediction minimising \eqref{Expectation_Brier} is $p_1=q_1$. We next go through two key cases
\begin{itemize}
\item For perfect prediction, the expectation value is given by
     \begin{equation}\label{Expectation_Brier_perf}
f(q_1)=q_1-q_1^2.
\end{equation}
As examples, consider the following cases:
\begin{itemize}
\item if $q_1=1/2=p_1$, the expected Brier score is $f(1/2)=1/4$ corresponding to the maximum of \eqref{Expectation_Brier_perf},
\item if for $q_1=1/10=p_1$, $f(1/10)=9/100$.
\end{itemize}
The expectation of the Brier score as a function of $q_1$, when $p_1=q_1$, i.e., \eqref{Expectation_Brier_perf}, looks as follows:\\
\begin{figure}[h!]
\begin{center}
\begin{tikzpicture}[scale=4]
  \draw[->] (0, 0) -- (1, 0) node[right] {$p$};
  \draw[->] (0, 0) -- (0, 0.3) node[above] {$f(p)$};
  \draw (-0.01, 0.25) to (0.01, 0.25);
    \draw (0.5,-0.01) to (0.5,0.01);
        \draw (1,-0.01) to (1,0.01);
  \draw[scale=1, domain=0:1, smooth, variable=\p, blue] plot ({\p}, {\p-\p*\p});
    \node[below] at (1,0) {1};
        \node[below] at (0.5,0) {0.5};
  \node[left] at (0,0.25) {0.25};
\end{tikzpicture}
\end{center}
  \caption{Expectation of the Brier score for the optimal prediction as a function of the underlying true probability with $p_1=q_1$.}
  \label{fig:best_exp}
\end{figure}
\item Next we compare the expected Brier score under perfect prediction with $p_1=q_1$ to 
\begin{itemize} 
\item the expectation of the Brier score under perfect prediction but increased true probability, i.e.,  $\tilde{p}_1:=p_1+\varepsilon=\tilde{q}_1=q_1+\varepsilon$ for $q_1+\varepsilon\leq 1/2$, and taking the difference
     \begin{equation}\label{Expectation_diff1}
g(q_1+\varepsilon,q_1+\varepsilon)-g(q_1,q_1)=\varepsilon(1-2q_1-\varepsilon) 
\end{equation}
Note that the difference depends strongly on $q_1$, and that the same holds if we shift true probability and perfect prediction from $q_1>1/2$ to some value $q_1-\varepsilon$ for $q_1-\varepsilon\geq 1/2$.
\item to the expectation of the Brier score with same true probability but slightly wrong prediction $\tilde{p}_1:=p_1+\varepsilon, \tilde{q}_1 =q_1\in[0,1]$ and taking the difference
     \begin{equation}\label{Expectation_diff2}
g(q_1+\varepsilon,q_1)-g(q_1,q_1)=\varepsilon^2
\end{equation}
Note that the difference does not depend on $q_1$, and the same holds for $\tilde{p}_1:=p_1-\varepsilon$. The calculations are in Appendix $\S$ \ref{exp_BS_diff}.
\end{itemize}
\item From the previous point, we conclude that if we slightly predict wrong, the costs are almost inexistent, e.g. if the true probability is $q_1$ but we predicted $p_1=q_1+0.1$, the expected difference between the Brier score of the perfect prediction and the slightly wrong one is $0.01$ by \eqref{Expectation_diff2}. However, if the true probability is changed towards $0.5$, e.g. from $q_1=0.1$ to $\tilde{q}_1=0.2$ and we have perfect prediction, the expected difference of the Brier scores is $0.07$ by \eqref{Expectation_diff1} which is seven times more.
\end{itemize}

 \item  $n$ arbitrary: As the Brier score is a rescaled sum of the Brier score with $n=1$, the above conclusion extend in a straightforward way. We note the following nice order preservation property from \eqref{Expectation_diff2}: If $p_1,\tilde{p_1}\in[0,1]$ such that $|p_1-q_1|< |\tilde{p_1}-q_1|$, then 
$$ \E_{Y_1\sim Bern(q_1)}[BS(p_1, Y_1)]< \E_{Y_1\sim Bern(q_1)}[BS(\tilde{p_1}, Y_1)].
$$
This property is known as being effective \cite{spher_effective}, and for $n>1$ the order that is preserved is the order implied by the $l_2$ distance from $q\in[0,1]^n$ to the prediction $p\in[0,1]^n$.
\end{itemize}
    \item\label{BS_dependence} \textbf{Dependence of expectation and distribution of Brier score on true probabilities:} 
The Brier score for multiple observations, as defined in \eqref{Brier}, is the mean of the individual (one-dimensional) Brier scores. Its expectation and distribution depend on the true outcome probabilities $q_i$, which are generally unknown.

As illustrated by \eqref{Expectation_Brier_perf} and Figure \ref{fig:best_exp}, the expected Brier score under perfect predictions varies with the distribution of the $q_i $. For example, if the $ q_i $ are mostly concentrated near 0 or 1, the expected Brier score is close to 0. In contrast, if the $ q_i $ cluster around 0.5, the expected Brier score under perfect predictions approaches 0.25.

    \item\label{BS_unobservable} \textbf{Unobservability of true probabilities in clinical data:} In practice, the true probability of an event occurring for an individual patient is not observable. Each patient is unique, and we only observe whether the event occurs or not (i.e., a binary outcome). This is in contrast to simulation settings where one can compare the true probabilities to estimated probabilities of regression or ML models, i.e. as done here \cite{van_der_ploeg_modern_2014,hoessly_study_prot}.
\end{enumerate}
While individual-level probabilities will remain unknown, we can approximate their average value across a population by calculating the mean of observed outcomes overall, or in similar patient groups. This can be used to assess model calibration in clinical settings via calibration in the large (CIL) \cite{ICI_co,van_calster_calibration_2019}
    \begin{equation}\label{CIL}
    CIL(p, y) = \frac{1}{n} \sum_{i=1}^{n} p_i -\frac{1}{n} \sum_{i=1}^{n} y_i .
\end{equation}

\subsection{Notable mathematical features about the Brier score} \label{prop_BS_math} 

Below, we summarise basic mathematical properties of the Brier score for the understanding of the reader. Mathematical details and arguments are given in Appendix \ref{math_exp}.  

Consider the Brier score of $n$ outcomes, where each $Y_i\sim Bern(q_i)$. Let the expected incidence be denoted by $\bar{q}$, i.e. $\bar{q}:=\frac{\sum_{i=1}^nq_i}{n}$. 
\begin{enumerate}  [label=\textbf{(M\Roman*)}]
\item\label{BS_fix_noninf} \textbf{Brier score of non-informative model that uses prevalence as prediction, i.e. $p=(\bar{y},\cdots,\bar{y})$:}
The Brier score of \eqref{Brier} with the non-informative mean as predictor $\bar{y}_v=(\bar{y},\cdots ,\bar{y})$, is given by $\bar{y}-\bar{y}^2$, i.e.
$$
BS(\bar{y}_v,y)= \bar{y}-\bar{y}^2.
$$
\item\label{BS_bound} \textbf{Bound on expectation of perfect prediction, i.e. for $p=(q_1,\cdots, q_n)$:}
The expected value of the average Brier score of \eqref{Brier} with the true probabilities as predictors, i.e., $p=(q_1,\cdots ,q_n)$, is bounded above by $\bar{q}-\bar{q}^2$, i.e.,
$$
\E[BS((q_1,\cdots ,q_n),(Y_1,\cdots,Y_n))]\leq \bar{q}-\bar{q}^2,
$$
with equality if and only if $q_i=\bar{q}$ for all $i\in[n]$. 

\item\label{BS_CLT} \textbf{Typical Brier score with perfect prediction for large $n$:}
Assume the true probabilities $q_i$ itself are realisations of random variables $Q_i\sim F$. Then, by the law of large numbers(LLN) the Brier score \eqref{Brier} for $n$ big roughly equals $\E[BS(Q_1,Y_1)]$, and probabilities for deviation from this value can be calculated via the central limit theorem (CLT). Hence $\E[BS(Q_1,Y_1)]$ roughly equals $BS((q_1,\cdots ,q_n),(y_1,\cdots,y_n))]$ for $n$ large, and $\bar{y}$ roughly equals $\bar{q}$. Similarly if the predictions are given by $Q_1$ plus some iid error, the LLN and CLT appliy, and, e.g., tails can be analysed via large deviations theory. More detail on this perspective is in Appendix $\S$ \ref{LLN_CLT}.

\end{enumerate}  

\section{Why we care about the expectation: law of large numbers and central limit theorem}  \label{LLN_CLT}
Another way to understand the behavior of the Brier score for perfect predictions is through basic mathematical tools. Consider the setting of the simulations considered, where the true probability $q_i$ itself is a realisation of a random variable $Q_i\sim F$. Denote by $J_n=(Q_1,\cdots,Q_n)$ the random vector of the first $n$ true probabilities and $Z_n=(Y_1,\cdots,Y_n)$ the realisation of the corresponding $n$ Bernoulli random variables. We can consider the Brier score of the optimal predictor $J_n$ as
$$BS(J_n,Z_n)\xrightarrow[]{n\to\infty}\mathbb{E}[(Q_1 - Y_1)^2]$$
For sufficiently large $n$, the observed Brier score provides a stable estimate of its expectation by the LLN \cite{georgii2008stochastics}.

Furthermore, the individual terms $(Q_i - Y_i)^2$ have finite variance and hence by the CLT the distribution of the Brier score, when properly normalised, approaches a normal distribution.
\begin{equation}
    \sqrt{n} \left( BS(J_n,Z_n) - \mathbb{E}[(Q_1 - Y_1)^2] \right) \xrightarrow[]{d} \mathcal{N}(0, \sigma^2),
\end{equation}

where $\sigma^2$ represents the variance of $(Q_1 - Y_1)^2$, and $\xrightarrow[]{d}$ indicated convergence in distribution  \cite{georgii2008stochastics}. This asymptotic normality allows for understanding the asymptotic behaviour of perfect predictions in the context of the simulations.
\section{Connections to some other scores}\label{connections}
We mention three other scores connected to Brier score. The Brier score from \eqref{Brier}  equals mathematically the MSE, hence its square root is the root mean squared error (RMSE):
\begin{equation}\label{RMSE}
    RMSE(p,y):=\sqrt{BS(p, y)}.
\end{equation}
As the square root on $[0,1]$ is order preserving, i.e., if $a,b\in[0,1], a\leq b$ then $\sqrt{a}\leq\sqrt{b}$, and bijective, RMSE also takes values in $[0,1]$ and most of the observations we made apply sililarly to RMSE.

Two similar and often-used scores are the mean-absolute error (MAE) that is defined as
\begin{equation}\label{MAE}
        MAE(p, y) = \frac{1}{n} \sum_{i=1}^{n} |p_i - y_i|,
\end{equation}
and the CIL \cite{steyerberg2019clinical}
\begin{equation}\label{CIL_2}
        CIL(p, y) = \frac{1}{n} \sum_{i=1}^{n} p_i - y_i.
\end{equation}
These relate to Brier score (or RMSE) through the following inequalities \cite{Amann_Escher1}
\begin{equation}\label{inequality1}
        CIL(p, y)\leq    MAE(p, y)\leq    RMSE(p, y)
\end{equation}
as well as 
\begin{equation} \label{inequality2}
        MSE(p, y)\leq    MAE(p, y)\leq     RMSE(p, y).
\end{equation}

\section{More detail for mathematical understanding Brier score}

\subsection{Expectation of Brier score in one dimension}\label{exp_BS}
    Consider $Y_1\sim Bern(q_1),q_1\in[0,1]$ and one prediction $p_1\in[0,1]$, then the expected Brier score is given as 
   $$
    \E[BS(p_1, Y)] =p_1^2 - 2p_1q_1 + q_1.
$$
In order to derive the above formula, we can proceed as follows.The expected Brier score is 

\[
\E[BS(p_1, Y_1)] = \E[(p_1 - Y_1)^2].
\]

We can expand the terms to get

\[
\E[(p_1 - Y_1)^2] = \E[p_1^2 - 2p_1Y_1 + Y_1^2].
\]

Since \( p_1,p_1^2 \) are constants, we get:

\[
\E[p_1^2 - 2p_1Y_1 + Y_1^2] = p_1^2 - 2p_1\E[Y_1] + \E[Y_1^2].
\]

For a Bernoulli-distributed variable \( Y_1 \):

\[
\E[Y_1] = q_1, \quad \E[Y_1^2] = \E[Y_1] = q_1.
\]

Hence the expected Brier score is:

\[
\E[BS(p_1, Y_1)] = p_1^2 - 2p_1q_1 + q_1.
\]

\subsection{Differences in expectation of Brier score in one dimension}\label{exp_BS_diff}
Recall that
\[
g(p_1, q_1) = p_1^2 - 2p_1q_1 + q_1.
\]
\begin{itemize}
\item
To derive the first result, we compute:
\[
 g(q_1+\varepsilon, q_1+\varepsilon)-g(q_1, q_1)
= q_1^2 + 2q_1\varepsilon + \varepsilon^2 - 2(q_1^2 + 2q_1\varepsilon + \varepsilon^2) + q_1 + \varepsilon-q_1+q_1^2
\]
Thus,
\[
g(q_1, q_1) - g(q_1+\varepsilon, q_1+\varepsilon) = \varepsilon(1-2q_1-\varepsilon) 
\]
\item For the second expression,
\[
  g(q_1+\varepsilon, q_1)-g(q_1, q_1)
= -q_1^2 - 2q_1\varepsilon - \varepsilon^2 + 2q_1^2 + 2q_1\varepsilon - q_1
\]
\[
= q_1^2 - \varepsilon^2 - q_1
\]
Thus, the difference:
\[
 g(q_1+\varepsilon, q_1)-g(q_1, q_1)  = \varepsilon^2
\]
\end{itemize}
\subsection{Understanding the expectation of the general Brier score} \label{math_exp} 
In case $p_1=q_1$, the expectation value of Brier score is given by
     \begin{equation}
f(p_1)=p_1-p_1^2.
\end{equation}
This is a strictly concave function, meaning that for any $\alpha\in[0,1]$ and any $x,y\in[0,1]$,
\begin{equation}\label{concave}
f((1-\alpha)x+\alpha y)\geq (1-\alpha) f(x)+\alpha f(y)
\end{equation}
Now we come back to the case where we have $ n$ observations and we consider the Brier score. 
\begin{lemma}
Consider the Brier score of $n$ outcomes, where $\frac{\sum_{i=1}^ny_i}{n}=\bar{y}$. Then the Brier score of \eqref{Brier} with prevalence as predictors, i.e., $(p_1,\cdots ,p_n)=(\bar{y},\cdots ,\bar{y})$ equals $\bar{y}-\bar{y}^2$, i.e.
$$
BS((\bar{y},\cdots ,\bar{y}),(y_1,\cdots,y_n))=\bar{y}-\bar{y}^2
$$
\end{lemma}
\begin{proof}
Let $n=a+b$ such that $\bar{y}=\frac{a}{a+b}$, and Brier score is given as
$$
\frac{1}{a+b}(a(\frac{a}{a+b}-1)^2+b(\frac{a}{a+b})^2)
$$
which we rewrite as 
$$
\frac{a}{a+b}(\frac{a}{a+b}-1)^2-(\frac{a}{a+b}-1)(\frac{a}{a+b})^2)=\bar{y}(\bar{y}-1)^2+(1-\bar{y})\bar{y}^2=\bar{y}-\bar{y}^2
$$
\end{proof}

\begin{lemma}
Consider the Brier score of $n$ outcomes, where $\frac{\sum_{i=1}^nq_i}{n}=\bar{q}$. Then the expected value of the average Brier score of \eqref{Brier} with the true probabilities as predictors, i.e., $(p_1,\cdots ,p_n)=(q_1,\cdots ,q_n)$, is bounded above by $\bar{q}-\bar{q}^2$, i.e.
$$
\E[BS((q_1,\cdots ,q_n),(Y_1,\cdots,Y_n))]\leq \bar{q}-\bar{q}^2,
$$
with equality if and only if $q_i=\bar{q}$ for all $i\in[n]$. 
\end{lemma}
\begin{proof}
Let $n\geq1$. Then
$$\E[BS((q_1,\cdots ,q_n),(Y_1,\cdots,Y_n))]=\frac{1}{n}(\sum_{i=1}^n q_i-q_i^2).$$
We can rewrite this as
$$=\frac{1}{n}\sum_{i=1}^n q_i-\frac{1}{n}\sum_{i=1}^n q_i^2=\bar{q}-\frac{1}{n}\sum_{i=1}^n q_i^2.$$
As the function $x\to x^2$ is strictly convex, we can apply Jensens Inequality to get 
$$\frac{1}{n}\sum_{i=1}^n q_i^2\geq (\frac{1}{n}\sum_{i=1}^n q_i)^2=\bar{q}^2,$$
such that finally we can bound it as
$$\bar{q}-\frac{1}{n}\sum_{i=1}^n q_i^2\geq \bar{q}-\bar{q}^2.$$
The statement with equality if and only if  $q_i=\bar{q}$ for all $i\in[n]$ also follows from Jensen.
\end{proof}

Another simple observation, with proof here is the following
\begin{lemma}
Consider the Brier score of $n$ outcomes, where $\frac{\sum_{i=1}^nq_i}{n}=c$. Then the expected value of the average Brier score of \eqref{Brier} with the non-informative mean as predictors, i.e., $(p_1,\cdots ,p_n)=(c,\cdots ,c)$, is given by $c-c^2$, i.e.
$$
\E[BS((c,\cdots ,c),(Y_1,\cdots,Y_n))]= c-c^2
$$
\end{lemma}
\begin{proof}
Using \eqref{Expectation_Brier} and \eqref{Brier} we get
$$
\E[BS((c,\cdots ,c),(Y_1,\cdots,Y_n))]=\frac{1}{n}\sum_{i=1}^n c^2 - \frac{1}{n}\sum_{i=1}^n 2cq_i + \frac{1}{n}\sum_{i=1}^n q_i
$$
which we can simplify using $\frac{\sum_{i=1}^nq_i}{n}=c$ to get
$$
c^2 - 2c\frac{1}{n}\sum_{i=1}^n q_i + \frac{1}{n}\sum_{i=1}^n q_i = c^2 - 2c^2+ c=c-c^2,
$$
which is what we wanted to show.
\end{proof}

\section{Mathematical proof impossibility of Brier score 0}\label{math_proof}
Recall the assumption.
\begin{assumption}\label{ass}Assume at least one of the true probabilities $q_i$ are in $(0,1)$.
\end{assumption}
\begin{lemma}
Let $n\in\N_{\geq 1}$, and let $y=(y_1,\cdots,y_n)$ be a realisation of a sequence of independent random variables, where $Y_i\sim Bern(q_i)$. If assumption \ref{ass} holds, the Brier score of the perfect prediction $p_{perf}=(q_1,\cdots, q_n)$ is bigger than zero, i.e.,
$$BS(p_{perf},y)>0.$$
\end{lemma}
\begin{proof}
Denote by $q=(q_1,\cdots,q_n)$ the vector of true probabilities, which also determines the perfect prediction vector $p_{perf}=q$. Assume we reordered them such that $q_1$ is in $(0,1)$, which holds by assumption assumption \ref{ass}. Define the following constant
$$\varepsilon:=\min\{q_1,1-q_1\}$$
By reordering and assumption, $\varepsilon>0$.
We can bound $BS(p_{perf},y)$ from below as follows
$$0<\frac{\varepsilon^2}{n}\leq BS(p_{perf},y).$$
\end{proof}

\newpage
\section{More details for the simulation via the ADEMP framework}\label{ADEMP}
We give more details on the simulations used in ADEMP framework \cite{ADEMP}. It took roughly 1h to run it on a Mac Studio M2 Max 2023.
\subsection{Aims}
\begin{itemize}
    \item The main aim is to compare Brier score across different settings both for true probabilities as well as predictions, which are based on the perfect prediction with some noise or bias added through median, quantile, and violin plots.
    \item A secondary aim is to compare Brier score to the mean incident Brier score, as well as to compare the CIL. In particular to estimate the probability that $\bar{y}-\bar{y}^2>BS_{perf}$, and to visualise the observed distribution of $\bar{y}-\bar{y}^2-BS_{perf}$ in different settings.
\end{itemize}

\subsection{Data Generating Mechanism}

\subsubsection{$y_i$ entered in the Brier score}
We sample true values for $q_i$ under some distributions which are subsequently used to simulate $Y_i \sim Bern(q_i)$. The sample distribution for $q_i$ are based on the following:
\begin{itemize}
    \item $Unif(a,b)$, where $(a,b) \in \{(0,1),(0,0.2)\}$.
    \item $Beta(\alpha,\beta)$, where $(\alpha,\beta) \in \{(2,5),(5,5),(3,3)\}$.
    \item Osteoporosis: Logistic regression model based on NHANES 2007/2008 \cite{NHANES} data using the nhanesA-package \cite{nhanesA} with complete case analysis; $7\%$ have osteoporosis.
    \begin{itemize}
        \item Outcome: Osteoporosis.
        \item Predictors: Vitamin D, calcium, weight, height, smoking, number of persons in household, age, US citizen status, education, gender. The following were considered nonlinear via rcs spline transformation with rms R-package \cite{rms_package} with 3 default knots:Vitamin D, calcium, age.
    \end{itemize}
    \item Smoking: Logistic regression model based on NHANES 2007/2008 \cite{NHANES} data using the nhanesA-package \cite{nhanesA} with complete case analysis; $26.3\%$ do smoke. 
    \begin{itemize}
        \item Outcome: Smoking.
        \item Predictors: Vitamin D, calcium, bmi, osteoporosis, number of persons in household, age, US citizen status, education, gender. The following were considered nonlinear via rcs spline transformation with rms R-package \cite{rms_package} with 3 default knots:Vitamin D, calcium, bmi, age.
    \end{itemize}
\end{itemize}

\subsubsection{Predictions $p_i$ entered in the Brier score}
The predictions entered are functions of the $q_i$, and the following settings are considered:
\begin{itemize}
    \item Perfect: $p_i=q_i$, i.e., perfect predictions.
    \item $+0.1$: $p_i=q_i+0.1$, i.e., slightly biased predictions.
    \item $+Unif(-0.1,0.1)$ resp $+Unif(-0.05,0.05)$: $p_i=q_i+X_i$, where $X_i\sim Unif(-0.1,0.1)$ resp $Unif(-0.05,0.05)$, i.e., disturbed but unbiased predictions.
   \item $+(1-2Bern(1/2))\cdot 0.1$: $p_i=q_i+(1-2X_i)\cdot 0.1$, where $X_i\sim Bern(1/2)$, i.e., slightly disturbed but unbiased predictions.

\end{itemize}
In case the $p_i$ are smaller than zero then $p_i$ is set to zero, and if they are bigger than one set to one.

\subsubsection{Sample Size}
Sample sizes considered are $n=\{300,1000\}$.

\subsubsection{Number of DGM Scenarios and Simulation Runs}
\begin{itemize}
    \item $|\#\text{ options for  dist.Y}| \cdot |\#\text{ options for  dist. p}|= 7 \cdot 5 = 35$ scenarios.
    \item $N=5000$ simulation repetitions per scenario.
\end{itemize}

\subsection{Estimand/Target of Analysis}
\begin{itemize}
    \item Distribution, median, quantiles of observed Brier score.
\end{itemize}

\subsection{Methods}
\subsubsection{Basis of Simulations}
The simulation is run as a Monte-Carlo simulation in R, where for the two settings based on NHANES data, the $q_i$ are based on the predicted value of the logistic regression for the corresponding observation, and the $q_i$ are subsampled without replacement. For the other distributions the $q_i$ are sampled iid from the distribution.

\subsection{Performance Measures}
\begin{itemize}
    \item distribution, median, 5\% and 95\% quantile Brier score in violin plot.
        \item estimate the probability that $\bar{y}-\bar{y}^2>BS_{perf}$.
    \item violin plot of $\bar{y}-\bar{y}^2-BS_{perf}$.
\end{itemize}

\end{document}